\documentclass[11pt]{article}
\usepackage[utf8]{inputenc}
\usepackage{float}
\usepackage{mathtools}
\usepackage{amsmath}
\usepackage{amsthm}
\usepackage{amssymb}
\usepackage{geometry}
\makeatletter

\floatstyle{ruled}
\newfloat{algorithm}{tbp}{loa}
\providecommand{\algorithmname}{Algorithm}
\floatname{algorithm}{\protect\algorithmname}

\@ifundefined{date}{}{\date{}}
\usepackage{fullpage}
\usepackage{algorithmic}
\usepackage{hyperref}

\makeatother

\theoremstyle{plain}
\newtheorem{thm}{\protect\theoremname}
\newtheorem{lem}[thm]{\protect\lemmaname}
\newtheorem{prop}[thm]{\protect\propositionname}
\newtheorem{cor}[thm]{\protect\corollaryname}
\providecommand{\corollaryname}{Corollary}
\providecommand{\lemmaname}{Lemma}
\providecommand{\propositionname}{Proposition}
\providecommand{\theoremname}{Theorem}

\begin{document}
\global\long\def\R{\mathbb{R}}%

\global\long\def\E{\mathbb{E}}%

\global\long\def\Var{\mathrm{Var}}%

\global\long\def\Prob{\mathbb{P}}%

\global\long\def\tr{\operatorname{tr}}%

\global\long\def\clip{\operatorname{clip}}%

\global\long\def\wtO{\widetilde{O}}%

\global\long\def\whSigma{\widehat{\Sigma}}%

\global\long\def\eps{\epsilon}%

\global\long\def\tstop{t_{\mathrm{stop}}}%

\allowdisplaybreaks
\title{Gap-free Differentially Private PCA for Gaussian Data}
\author{
  Alina Ene\thanks{Department of Computer Science, Boston University, \texttt{aene@bu.edu}.} \and
  Huy L. Nguyen\thanks{Khoury College of Computer and Information Science, Northeastern University, \texttt{hu.nguyen@northeastern.edu}.}
}
\date{}
\maketitle
\begin{abstract}
We give a gap-free $(\epsilon,\delta)$-differentially private algorithm
for the principal component analysis (PCA) problem with Gaussian data.
The algorithm is based on a private variant of the power iteration
method, and it is computationally efficient. 
\end{abstract}

\section{Introduction}

Dimension reduction using principal component analysis is a classical
problem in data analysis with many applications. In certain applications
involving sensitive data, it is important to protect the privacy of
the users. Differential privacy has emerged as the gold standard for
privacy protection and it is widely adopted in both private enterprises
and government entities. In this work, we study the problem of finding
the approximately best-fit one-dimension subspace, arguably the most
basic form of PCA, with differential privacy. This problem has been
studied in many prior works including \cite{DBLP:conf/stoc/HardtR12,DBLP:conf/stoc/HardtR13,DBLP:conf/nips/HardtP14,DBLP:conf/nips/LiuK0O22,liu2022differential,DBLP:conf/stoc/DworkTT014}
and many others. We consider the formulation where the data is generated
i.i.d. from an unknown Gaussian distribution and the goal is to recover
an approximately best fit one dimension subspace of the distribution
while satisfying differential privacy for general data.

Formally, we assume we are given a dataset consisting of $n$ datapoints
$y_{1},\dots,y_{n}\in\R^{d}$ where each datapoint is the private
data of a user. We consider replacement adjacency, i.e., two datasets
are neighboring if one is obtained from the other by replacing one
datapoint. Our goal is to design an $(\epsilon,\delta)$-differentially
private algorithm that ensures privacy for all inputs and achieves
good utility for inputs that are sampled i.i.d. from a Gaussian distribution
with an unknown covariance matrix. More precisely, for the utility
analysis, we assume $y_{1},\dots,y_{n}\stackrel{\mathrm{iid}}{\sim}N(0,\Sigma)$
where $\Sigma=S^{2}\succeq0$ is unknown. Let $s_{1}:=\|S\|_{2}$
and $\lambda_{1}:=\|\Sigma\|_{2}=s^{2}_{1}$. Our target utility guarantee
is that, for any approximation factor $\alpha\in(0,1)$ given as input,
with constant probability, the algorithm outputs a unit vector $x$
satisfying
\[
x^{\top}\Sigma x\ge(1-\alpha)\lambda_{1}.
\]
Liu et al. \cite{liu2022differential} give a computationally inefficient
algorithm that constructs an $\alpha$-approximate solution using
$n=\widetilde{\Theta}\left(\frac{d}{\alpha^{2}}+\frac{d}{\epsilon\alpha}\right)$
samples, where the $\widetilde{O}/\widetilde{\Omega}/\widetilde{\Theta}$
notation suppresses factors that are logarithmic in $d$, $1/\alpha$,
$1/\epsilon$, and $1/\delta$. Here we give a computationally efficient
algorithm that attains this sample complexity. 

For setting the clipping radius and noise scale of our algorithm,
we assume that a public constant-factor approximation of $\lambda_{1}$
is available, i.e. we assume we are given a public value $\Lambda$
satisfying
\begin{equation}
\lambda_{1}\le\Lambda\le2\lambda_{1}.\label{eq:scale}
\end{equation}
This assumption is for ease of exposition. If there is a known finite
range for $\Lambda$, one can guess this value from powers of two
and use the exponential mechanism to select the best solution from
all the guesses. Alternatively, without any known finite range, one
can use $\wtO(\log(1/\delta)/\eps)$ groups of $\wtO(d)$ samples
each to populate a private histogram of the top eigenvalues of each
group. Then we use the private histogram algorithm by \cite{KarwaVadhan2018}
to select the power of two closest to the most number of top eigenvalues.
We omit the details for these approaches to obtain $\Lambda$ privately.
\begin{thm}
\label{thm:main} Suppose that $0<\alpha<c$ where $c$ is a sufficiently
small absolute constant, $0<\epsilon<1$, and $0<\delta<1/10$. Let
$\Lambda$ be a public value. Algorithm \ref{alg:Algorithm-clipped-pca}
is $(\epsilon,\delta)$-differentially private on all inputs. If $\Lambda$
satisfies $\lambda_{1}\leq\Lambda\leq2\lambda_{1}$ and the datapoints
are sampled from the Gaussian distribution, with constant probability,
the algorithm outputs a unit vector $x$ satisfying 
\[
x^{\top}\Sigma x\ge(1-\alpha)\lambda_{1},
\]
provided that the number of samples is
\[
n=\widetilde{\Theta}\left(\frac{d}{\alpha^{2}}+\frac{d}{\epsilon\alpha}\right).
\]
\end{thm}

\textbf{Related work. }As noted above, Liu et al. \cite{liu2022differential}
give a computationally inefficient gap-free algorithm with the sample
complexity stated above. In the setting where $\Sigma$ has a positive
eigenvalue gap $\lambda_{1}-\lambda_{2}>0$, Liu et al. \cite{DBLP:conf/nips/LiuK0O22}
and Nguyen et al. \cite{NguyenEN26} give algorithms with error guarantees
and sample complexities that depend on the gap. The algorithms of
Liu et al. \cite{liu2022differential} and Liu et al. \cite{DBLP:conf/nips/LiuK0O22}
are for the more general sub-Gaussian distributions. There has also
been an extensive line of work under varying models of privacy for
worst-case data with utility guarantees that are parameterized by
the coherence of the data matrix \cite{DBLP:conf/stoc/HardtR12,DBLP:conf/stoc/HardtR13,DBLP:conf/nips/HardtP14,DBLP:conf/colt/BalcanDWY16,NicolasSMMC25,dOrsiN26,NguyenEN26}.
When the datapoints are sampled from a Gaussian distribution, the
data matrix is known to have low coherence (see e.g. \cite{NguyenEN26}).
Several of these algorithms are for the setting where the input data
matrix whose rows are the datapoints $y_{i}$ has a positive singular
value gap. We refer the reader to \cite{NguyenEN26} for a detailed
discussion of these works, and to \cite{liu2022differential} for
a discussion of other related work.

\textbf{Acknowledgements}. The proofs were developed with GPT 5.6
Sol and GPT 6 Astra on September 14, 2026 based on our prior ideas
with adaptive thresholding for non-random data \cite{NguyenEN26}
and martingale analysis for stochastic gradient descent \cite{LiuNNEN23}
and they were edited and rearranged by the authors. We verified and
rewrote all the proofs. We thank Gavin Brown for letting us know about
the problem at the Workshop on the Intersections of Differential Privacy
and Sublinear Algorithms. We also acknowledge a concurrent work by
Gu et al. \cite{GuKumarTianYang2026}, which also settled this problem
posed by Gavin Brown and appeared on arXiv a few days before our manuscript.
Their work also uses privacy to justify adaptive clipping but appears
to use a more global technique rather than our single step analysis,
which is more reminiscent of the proof for stochastic gradient descent.

\section{Preliminaries}

We first note a few useful facts regarding Gaussian concentration.
The following lemma gives the concentration of the length of a $d$-dimensional
standard Gaussian vector.
\begin{lem}
\cite[Section 4.1]{LaurentMassart2000}\label{lem:chisq} If $g\sim N(0,I_{d})$,
\[
\Prob[\left\Vert g\right\Vert ^{2}-d\ge2\sqrt{du}+2u]\le e^{-u},
\]
and 
\[
\Prob[d-\left\Vert g\right\Vert ^{2}\ge2\sqrt{du}]\le e^{-u}.
\]
 
\end{lem}

The next lemma bounds the correlation between a fixed unit vector
$x\in\R^{d}$ and a random Gaussian vector with covariance $\Sigma$
in terms of $\lambda_{1}=\left\Vert \Sigma\right\Vert _{2}$ and the
dimension $d$.
\begin{lem}
\label{lem:one-sample} Let $y\sim N(0,\Sigma)$. For every fixed
unit vector $x\in\R^{d}$ and every $u\ge1$, 
\begin{equation}
\Prob\left[\|yy^{\top}x\|>\lambda_{1}(\sqrt{d}+\sqrt{2u})\sqrt{2u}\right]\le3e^{-u}.\label{eq:one-sample}
\end{equation}
\end{lem}

\begin{proof}
First, note that $\|yy^{\top}x\|=\|y\|\cdot|y^{\top}x|$. Write 
\[
y=Sg,\qquad g\sim N(0,I_{d}).
\]
Then 
\[
\|y\|\le s_{1}\|g\|.
\]
Lemma \ref{lem:chisq} implies 
\[
\Prob[\|y\|>s_{1}(\sqrt{d}+\sqrt{2u})]\le e^{-u}.
\]
Moreover, 
\[
y^{\top}x\sim N(0,x^{\top}\Sigma x),
\]
and 
\[
x^{\top}\Sigma x\le\lambda_{1}=s^{2}_{1}.
\]
Thus 
\[
\Prob[|y^{\top}x|>s_{1}\sqrt{2u}]\le2e^{-u}.
\]
The result then follows by a union bound. 
\end{proof}

The next lemma shows that the empirical covariance is close to the
distributional covariance when $n\gg d$. It follows from standard
results for random matrices \cite[Theorem II.13]{DavidsonSzarek01}. 
\begin{lem}
\label{lem:cov}Let $y_{1},\ldots,y_{n}\sim N(0,\Sigma)$ in $d$
dimensions with $n\ge d$. Let $\whSigma=\frac{1}{n}\sum^{n}_{i=1}y_{i}y^{\top}_{i}$.
With probability $1-\gamma$ for arbitrary constant $\gamma$, $\|\whSigma-\Sigma\|_{2}=O\left(\lambda_{1}\left(\sqrt{\left(d+\log(1/\gamma)\right)/n}+\left(d+\log(1/\gamma)\right)/n\right)\right)$. 
\end{lem}

\begin{proof}
Let $\Sigma=S^{2}$. Note that $y_{i}=Sg_{i}$ where $g_{i}\sim N(0,I)$.
We have $\whSigma-\Sigma=S(\frac{1}{n}\sum_{i=1}g_{i}g^{\top}_{i}-I)S$
so $\|\whSigma-\Sigma\|_{2}\le\lambda_{1}\text{\ensuremath{\left\Vert \frac{1}{n}\sum_{i=1}g_{i}g^{\top}_{i}-I\right\Vert }}_{2}$.
By \cite[Theorem II.13]{DavidsonSzarek01}, with probability $1-\exp(-nt^{2}/2)$,
we have $\text{\ensuremath{\left\Vert \frac{1}{n}\sum_{i=1}g_{i}g^{\top}_{i}-I\right\Vert }}_{2}\le a(2+a)$
where $a=\sqrt{d/n}+t$. The lemma follows from choosing $t=\Theta\left(\sqrt{\log(1/\gamma)/n}\right)$.
\end{proof}

\section{The clipped private power iteration}

For $R>0$, define
\[
\clip_{R}(z):=z\min\left\{ 1,\frac{R}{\|z\|}\right\} .
\]

\begin{algorithm}
\caption{Clipped private power iteration}

\label{alg:Algorithm-clipped-pca}

\begin{algorithmic}[1]

\STATE \textbf{Input}: Datapoints $y_{1},\ldots,y_{n}\in\R^{d}$,
approximate eigenvalue $\Lambda$,  privacy parameters $\epsilon,\delta\in(0,1]$

\STATE Let $\gamma=1/256,u=C\log\frac{nT}{\gamma}$, clipping radius
$R=\Lambda(\sqrt{d}+\sqrt{2u})\sqrt{2u}$, time $T=\frac{C\log(ed)}{\alpha}$
where $C$ is a sufficiently large constant, stricter privacy parameter
$\bar{\delta}=\min\left\{ \delta,\frac{\gamma}{10nT}\right\} .$

\STATE Sample $x_{0}$ uniformly from the unit sphere in $\R^{d}$

\STATE \textbf{for} $t=0\dots T-1$

\STATE $\qquad$$q_{t}=\frac{1}{n}\sum^{n}_{i=1}\clip_{R}(y_{i}y^{\top}_{i}x_{t})$

\STATE $\qquad$Sample $g_{t}\sim N(0,\tau^{2}I_{d})$ where $\tau=\frac{4R\sqrt{T\log(1/\bar{\delta})}}{n\eps}$

\STATE $\qquad$Let $z_{t}=q_{t}+g_{t}$ and $x_{t+1}=\frac{z_{t}}{\|z_{t}\|}$

\STATE \textbf{end for}

\STATE \textbf{Sample $J\sim\operatorname{Unif}\{0,\ldots,T-1\}$
and return} $x_{J}$

\end{algorithmic}
\end{algorithm}

The algorithm is described in Algorithm \ref{alg:Algorithm-clipped-pca}.
It is the standard power method with a few changes. First, in each
iteration, the contribution of each data point is clipped to radius
$R=\wtO(\Lambda\sqrt{d})$ and the Gaussian mechanism is applied to
average to protect privacy. Second, we return a random iterate as
opposed to the last one. 

\textbf{Analysis outline.} The privacy analysis is done in Section
\ref{sec:Privacy}. Section \ref{sec:DP-decoupling} shows that because
the algorithm is differentially private, the correlation between the
iterates $x_{t}$ and the datapoints $y_{i}$ is similar to the correlation
between $x_{t}$ and a fresh sample $y\sim N(0,\Sigma)$, which is
small with high probability. Thus, clipping does not happen with high
probability. Conditioned on clipping not happening, the utility analysis
is completed using any analysis of the noisy power method without
clipping. There are many such analyses in the literature such as \cite{AllenZhuLi2017,Shamir2016PCA}.
For completeness and with a relatively simple argument, we include
an analysis for noisy power method in Section \ref{sec:Noisy-power-method}.
Finally in Section \ref{sec:Proof-main-theorem}, we put these analyses
together to prove the main theorem. 

\section{Privacy}\label{sec:Privacy}

Conditioned on all preceding private outputs, $x_{t}$ is fixed. Every
summand in $q_{t}=\frac{1}{n}\sum^{n}_{i=1}\clip_{R}(y_{i}y^{\top}_{i}x_{t})$
has Euclidean norm at most $R$. Therefore replacement of one sample
changes $q_{t}$ by at most 
\[
\Delta_{2}(q_{t})\le\frac{2R}{n}.
\]

\begin{prop}
Privacy of the clipped iteration \label{prop:privacy} Let 
\[
\bar{\delta}\le\min\left\{ \delta,\frac{\gamma}{10nT}\right\} .
\]
It is sufficient to choose 
\begin{equation}
\tau=\frac{2R\sqrt{T\log(1/\bar{\delta})}}{cn\eps}.\label{eq:tau}
\end{equation}
Then the algorithm is $(\eps,\bar{\delta})$-DP.
\end{prop}

\begin{proof}
One iteration has sensitivity at most $\frac{2R}{n}$. Thus, each
iteration is $(\alpha,\frac{\alpha c^{2}\eps^{2}}{2T\log(1/\bar{\delta})})$-Renyi
DP for some constant $c$. By composition, the algorithm is $(\alpha,\frac{\alpha c^{2}\eps^{2}}{2\log(1/\bar{\delta})})$-RDP.
By the conversion to approximate DP, this algorithm is $\left(\frac{\alpha c^{2}\eps^{2}}{2\log(1/\bar{\delta})}+\frac{\log(1/\bar{\delta})}{\alpha-1},\bar{\delta}\right)$-DP.
With $\alpha=1+2\log(1/\bar{\delta})/\eps$, the algorithm is $\left(c^{2}\eps+\frac{c^{2}\eps^{2}}{2\log(1/\bar{\delta})}+\eps/2,\bar{\delta}\right)$-DP.
The proposition follows from the appropriate choice of $c\le1/2$.

Note that while our privacy parameter is stricter $\bar{\delta}\le\min\left\{ \delta,\frac{\gamma}{10nT}\right\} $
than provided, this is generally not a big issue in practice because
for a meaningful privacy guarantee, we need $\delta\ll1/n$.
\end{proof}

\section{DP decoupling and the no-clipping event}\label{sec:DP-decoupling}

In this section, we show that clipping does not happen with high probability.
We start by showing that the correlation between the iterates and
the samples is small. The preceding Lemma \ref{lem:one-sample} does
so for a fixed vector $x$. The iterate $x_{t}$, however, depends
on the same samples $y_{i}$. The next lemma decouples this dependence.
\begin{lem}
\label{lem:decouple} Let
\[
D=(y_{1},\ldots,y_{n})
\]
be i.i.d.\ from a distribution $P$, and let $M(D)$ be the output
of an $(\eps_{0},\delta_{0})$-DP algorithm with input $D$. Suppose
an event $E(y,z)$ satisfies 
\[
\sup_{z}\Prob_{y\sim P}[E(y,z)]\le p.
\]
Then 
\[
\Prob[E(y_{i},M(D))]\le e^{\eps_{0}}p+\delta_{0}
\]
for every $i$.
\end{lem}

\begin{proof}
Let $y_{i}'$ be an independent copy of $y_{i}$, and let $D_{-i}$
denote the remaining observations. For fixed $y$, write 
\[
E_{y}=\{z:E(y,z)\}.
\]
For any values of $D_{-i},y_{i},y_{i}'$, differential privacy gives
\[
\Prob[M(D_{-i},y_{i})\in E_{y_{i}}\mid D_{-i},y_{i},y_{i}']\le e^{\eps_{0}}\Prob[M(D_{-i},y_{i}')\in E_{y_{i}}\mid D_{-i},y_{i},y_{i}']+\delta_{0}
\]
Averaging over all variables $D_{-i},y_{i},y_{i}'$, we obtain
\[
\Prob_{D_{-i},y_{i},y_{i}'}[M(D_{-i},y_{i})\in E_{y_{i}}]\le e^{\eps_{0}}\Prob_{D_{-i},y_{i},y_{i}'}[M(D_{-i},y_{i}')\in E_{y_{i}}]+\delta_{0}
\]
The random variable $M(D_{-i},y_{i}')$ is independent of $y_{i}$.
Thus, conditioning on its value $z$, 
\[
\Prob_{y_{i}}[E(y_{i},z)]\le p.
\]
The result follows from combining the last two inequalities.
\end{proof}

Given the decoupling lemma and the fact that a fresh sample is unlikely
to correlate with the iterate by Lemma~\ref{lem:decouple} and \ref{lem:one-sample},
we next show that with high probability, the algorithm does not clip.
\begin{prop}
\label{prop:no-clipping} Let 
\[
u=C\log\frac{nT}{\gamma},
\]
and set 
\begin{equation}
R=\Lambda(\sqrt{d}+\sqrt{2u})\sqrt{2u}.\label{eq:R}
\end{equation}
Suppose the algorithm is $(\eps,\bar{\delta})$-DP, where $\eps\le1$
and $\bar{\delta}\le\frac{\gamma}{2nT}$. Then 
\[
\Prob\left[\exists i,t:\|y_{i}y^{\top}_{i}x_{t}\|>R\right]\le\gamma.
\]
Consequently,
\[
R=\wtO(\lambda_{1}\sqrt{d}).
\]
\end{prop}

\begin{proof}
Every prefix iterate $x_{t}$ is post-processing of a prefix of the
private transcript and is therefore itself $(\eps,\bar{\delta})$-DP.

For a fixed $x$, Lemma~\ref{lem:one-sample} and $\lambda_{1}\le\Lambda$
imply 
\[
\Prob_{y}[\|yy^{\top}x\|>R]\le3e^{-u}.
\]
Lemma~\ref{lem:decouple} therefore gives 
\[
\Prob[\|y_{i}y^{\top}_{i}x_{t}\|>R]\le3e^{\eps-u}+\bar{\delta}.
\]
Union bounding over $i\in[n]$ and $t<T$, 
\[
\Prob[\exists i,t:\|y_{i}y^{\top}_{i}x_{t}\|>R]\le3nTe^{\eps-u}+\underbrace{nT\bar{\delta}}_{\le\gamma/2}.
\]
For a sufficiently large constant $C$ in $u=C\log\frac{nT}{\gamma}$,
the right-hand side is at most $\gamma$. 
\end{proof}

\section{Noisy power method }\label{sec:Noisy-power-method}

Next, we analyze the algorithm when there is no clipping. A similar
result is shown by \cite{Shamir2016PCA} with $1/d$ success probability.
For completeness, we include a proof with a constant success probability.

Let $W\succeq0$ be fixed, let $\lambda=\|W\|_{2}>0$, and fix a unit
vector $v$ with $Wv=\lambda v$. Consider the iterative method that
starts with $x_{0}$ chosen uniformly at random from the unit sphere
and performs the following update for every $0\leq t<T$: 
\begin{equation}
x_{t+1}=\frac{Wx_{t}+g_{t}}{\|Wx_{t}+g_{t}\|},\qquad g_{t}\stackrel{\mathrm{iid}}{\sim}N(0,\tau^{2}I_{d})\label{eq:ideal-power}
\end{equation}

\begin{thm}
\label{thm:noisy} There are universal constants $c_{0}=2^{-14},C>0$
such that, for every integer $T\ge1$, if 
\begin{equation}
s:=\frac{d\tau^{2}T}{\lambda^{2}}\le c_{0},\label{eq:noise-assumptions}
\end{equation}
then with probability at least $1/16$, 
\begin{equation}
\frac{1}{T}\sum^{T-1}_{t=0}(\lambda-x^{\top}_{t}Wx_{t})\le C\left[\frac{\lambda\log(2d)}{T}+\sqrt{\frac{d\tau^{2}}{T}}+\frac{d\tau^{2}}{\lambda}\right].\label{eq:avg}
\end{equation}
For an independent uniform $J\in\{0,\ldots,T-1\}$, the same bound,
after increasing $C$, holds for $\lambda-x^{\top}_{J}Wx_{J}$ with
probability at least $1/32$.
\end{thm}

The intuition of the proof is as follows. Consider the potential function
$\Phi_{t}=\log\left(\left(v^{\top}x_{t}\right)^{2}\right)$. With
a random initialization, $\Phi_{0}$ is finite and because $v,x_{t}$
are unit vectors, we always have $\Phi_{t}\le0$. We can relate the
change in the potential in each step with the suboptimality $\Delta_{t}=\lambda-x^{\top}_{t}Wx_{t}$
of the solution $x_{t}$:
\[
\Phi_{t+1}-\Phi_{t}\ge\frac{\Delta_{t}}{\lambda}+\text{mean zero noise}-\text{quadratic noise}
\]
Telescoping this inequality gives
\[
\frac{1}{\lambda}\sum_{t}\Delta_{t}\le-\Phi_{0}+\sum_{t}\text{noise}
\]
The noise can be bounded using a martingale argument. This structure
is reminiscent of the proof for the convergence of stochastic gradient
descent.

First, we prove the key inequality mentioned above relating the change
in the potential to the suboptimality plus noise terms.
\begin{lem}
Let $r_{t}=\frac{v^{\top}g_{t}}{\lambda v^{\top}x_{t}},m_{t}=2r_{t}-\frac{2\langle Wx_{t},g_{t}\rangle}{\lambda^{2}}=2\left\langle \frac{v}{\lambda v^{\top}x_{t}}-\frac{Wx_{t}}{\lambda^{2}},g_{t}\right\rangle ,b_{t}=2r^{2}_{t}+\frac{\|g_{t}\|^{2}}{\lambda^{2}}$.
If $|r_{t}|\le1/2$, then 
\begin{equation}
\Phi_{t+1}-\Phi_{t}\ge\frac{\Delta_{t}}{\lambda}+m_{t}-b_{t}.\label{eq:power-drift}
\end{equation}
\end{lem}

\begin{proof}
Since $Wv=\lambda v$, 
\begin{align}
\Phi_{t+1}-\Phi_{t} & =2\log\frac{\left|v^{\top}x_{t+1}\right|}{\left|v^{\top}x_{t}\right|}\nonumber \\
 & =2\log\frac{\left|v^{\top}Wx_{t}+v^{\top}g_{t}\right|}{\|Wx_{t}+g_{t}\|\left|v^{\top}x_{t}\right|}\nonumber \\
 & =2\log\frac{\left|\lambda v^{\top}x_{t}+v^{\top}g_{t}\right|}{\|Wx_{t}+g_{t}\|\left|v^{\top}x_{t}\right|}\nonumber \\
 & =2\log\frac{\lambda\left|1+\frac{v^{\top}g_{t}}{\lambda v^{\top}x_{t}}\right|}{\|Wx_{t}+g_{t}\|}\nonumber \\
 & =2\log\left|1+r_{t}\right|-\log\frac{\|Wx_{t}+g_{t}\|^{2}}{\lambda^{2}}\label{eq:change-potential}
\end{align}
The first term above is bounded by $\log(1+r_{t})\ge r_{t}-r^{2}_{t}$
when $\left|r_{t}\right|\le1/2$. Next we bound the second term.

Because $0\preceq W\preceq\lambda I$, we have $W^{2}\preceq\lambda W$,
and hence 
\[
\frac{\|Wx_{t}\|^{2}}{\lambda^{2}}\le\frac{x^{\top}_{t}Wx_{t}}{\lambda}=1-\frac{\Delta_{t}}{\lambda}.
\]
Using $\log u\le u-1$ for $u>0$, we obtain 
\[
\log\frac{\|Wx_{t}+g_{t}\|^{2}}{\lambda^{2}}\le-\frac{\Delta_{t}}{\lambda}+\frac{2\langle Wx_{t},g_{t}\rangle}{\lambda^{2}}+\frac{\|g_{t}\|^{2}}{\lambda^{2}}.
\]
Substituting the bounds into Equation \ref{eq:change-potential} proves
the lemma.
\end{proof}

Next we show that the initial potential $\Phi_{0}$ is bounded with
constant probability.
\begin{lem}
For the random initialization on the sphere, we have
\begin{equation}
\Prob\left[\left(v^{\top}x_{0}\right)^{2}\ge\frac{1}{2d}\right]\ge1/12\label{eq:power-init}
\end{equation}
\end{lem}

\begin{proof}
By the rotational symmetry of the initialization, it suffices to consider
the case $v=e_{1}$, the first standard basis vector. A $d$-dimensional
gaussian vector $g$ is identically distributed as $Rx_{0}$ where
$R^{2}\sim\chi^{2}_{d}$ and $R,x_{0}$ are independent. Note that
$\Var[R^{2}]=2d$ and $\E[R^{2}]=d$ so $\E[R^{4}]=2d+d^{2}$. Considering
the first coordinate of $g$, we have
\[
3=\E[g^{4}_{1}]=\E[(Rx_{0,1})^{4}]=\E[R^{4}]\E[x^{4}_{0,1}]
\]
and thus $\E[x^{4}_{0,1}]=\frac{3}{2d+d^{2}}$. We also have $\E[x^{2}_{0,1}]=\frac{1}{d}$
by symmetry of the coordinates. By Paley-Zygmund's inequality,
\[
\Prob\left[x^{2}_{0,1}\ge\frac{1}{2d}\right]\ge\frac{1}{4}\frac{1/d^{2}}{3/\left(2d+d^{2}\right)}\ge\frac{1}{12}
\]
\end{proof}

Next, we are ready to prove Theorem \ref{thm:noisy}.
\begin{proof}[Proof of Theorem~\ref{thm:noisy}]
Fix any initialization with $\Phi_{0}\ge\log(1/2d)$. We define the
stopping time 
\[
\tstop:=\inf\left\{ t\in\{0,\ldots,T\}:\left(v^{\top}x_{t}\right)^{2}<\frac{1}{8d}\right\} ,
\]
with $\inf\varnothing=\infty$. Let $\left(\mathcal{F}_{t}\right)_{0\leq t\leq T}$
be the filtration where $\mathcal{F}_{t}=\sigma(x_{0},g_{0},\ldots,g_{t-1})$.
Define stopped variables $M_{t}=m_{t}$ and $B_{t}=b_{t}$ for $t<\tstop$,
and $M_{t}=B_{t}=0$ otherwise.

On $t<\tstop$, conditional on $\mathcal{F}_{t}$, we have
\[
\E[r^{2}_{t}\mid\mathcal{F}_{t}]=\frac{\tau^{2}}{\lambda^{2}\left(v^{\top}x_{t}\right)^{2}}\le\frac{8d\tau^{2}}{\lambda^{2}}.
\]
Hence
\[
\E[b_{t}\mid\mathcal{F}_{t}]=\E\left[2r^{2}_{t}+\frac{\|g_{t}\|^{2}}{\lambda^{2}}|\mathcal{F}_{t}\right]\le\frac{17d\tau^{2}}{\lambda^{2}}.
\]
 On $t<\tstop$, we also have 
\begin{align*}
\E[m^{2}_{t}\mid\mathcal{F}_{t}] & =\E\left[4\left\langle \frac{v}{\lambda v^{\top}x_{t}}-\frac{Wx_{t}}{\lambda^{2}},g_{t}\right\rangle ^{2}|x_{t}\right]\\
 & =4\tau^{2}\left\Vert \frac{v}{\lambda v^{\top}x_{t}}-\frac{Wx_{t}}{\lambda^{2}}\right\Vert ^{2}\\
 & \le4\tau^{2}\left(\frac{2\left\Vert v\right\Vert ^{2}}{\lambda^{2}\left(v^{\top}x_{t}\right)^{2}}+\frac{2\left\Vert Wx_{t}\right\Vert ^{2}}{\lambda^{4}}\right)\\
 & \le4\tau^{2}\left(\frac{16d}{\lambda^{2}}+\frac{2}{\lambda^{2}}\right)\le\frac{72d\tau^{2}}{\lambda^{2}}
\end{align*}
The inequality follows from $\|Wx_{t}\|\le\lambda$ and $1/\left(v^{\top}x_{t}\right)^{2}\le8d$.
Since the stopped variables vanish after stopping, for every $t<T$
we therefore have 
\[
\E[M_{t}\mid\mathcal{F}_{t}]=0,\qquad\E[M^{2}_{t}\mid\mathcal{F}_{t}]\le\frac{72d\tau^{2}}{\lambda^{2}},\qquad\E[B_{t}\mid\mathcal{F}_{t}]\le\frac{17d\tau^{2}}{\lambda^{2}}.
\]
Let $s=d\tau^{2}T/\lambda^{2}$. Let $S_{k}=\sum_{t<k}M_{t}$. Note
that
\[
\mathbb{E}[S^{2}_{k+1}\mid\mathcal{F}_{k}]=S^{2}_{k}+\mathbb{E}[M^{2}_{k}\mid\mathcal{F}_{k}]\ge S^{2}_{k}
\]
and thus $S^{2}_{k}$ is a submartingale. Additionally, 
\[
\E[S^{2}_{k}]=\sum_{t<k}\E\left[\E[M^{2}_{t}\mid\mathcal{F}_{t}]\right]\le72s.
\]
By Doob's inequality applied to the submartingale $S^{2}_{k}$ and
Markov's inequality,
\begin{align*}
\Prob\left[\max_{k\le T}S^{2}_{k}>\left(32\right)^{2}s\right] & \le\frac{72}{32^{2}},\\
\Prob\left[\sum_{t<T}B_{t}>128s\right] & \le\frac{17}{128}.
\end{align*}
If $s=0$, both sums vanish and no tail bound is needed. Thus the
event 
\begin{equation}
\max_{k\le T}\left|\sum_{t<k}M_{t}\right|\le32\sqrt{s}\text{ and }\sum_{t<T}B_{t}\le128s\label{eq:power-good}
\end{equation}
has probability at least $1-\frac{72}{1024}-\frac{17}{128}\ge3/4$
conditioned on the good initialization event \eqref{eq:power-init}.

Recall the assumption \ref{eq:noise-assumptions} that $s=\frac{d\tau^{2}T}{\lambda^{2}}\le c_{0}$.
Take $c_{0}=2^{-14}$. On the event \eqref{eq:power-good}, $32\sqrt{s}\le1/4$
and $128s\le1/4$. Every active step with $t<\tstop$ also has $2r^{2}_{t}\le B_{t}\le1/4$,
so $|r_{t}|<1/2$. Consequently \eqref{eq:power-drift} applies up
to and including the step that would cross the stopping threshold.
If $\tstop\le T$, telescoping it for $t<\tstop$ gives 
\[
\Phi_{\tstop}\ge\Phi_{0}-\frac{1}{2},\qquad\left(v^{\top}x_{\tstop}\right)^{2}\ge e^{-1/2}\left(v^{\top}x_{0}\right)^{2}>\frac{1}{8d},
\]
contradicting the definition of $\tstop$. Thus $\tstop>T$ throughout
the good event \eqref{eq:power-good}.

On the intersection of the events \eqref{eq:power-init} and \eqref{eq:power-good},
no stopping occurs. Summing \eqref{eq:power-drift} and using $\Phi_{T}\le0$
yields 
\[
\frac{1}{\lambda}\sum_{t<T}\Delta_{t}\le-\Phi_{0}+32\sqrt{s}+128s\le\log(2d)+32\sqrt{s}+128s=\log(2d)+32\sqrt{d\tau^{2}T/\lambda^{2}}+128d\tau^{2}T/\lambda^{2}
\]
Multiplication by $\lambda/T$ proves \eqref{eq:avg}. 
\[
\frac{1}{T}\sum_{t<T}\Delta_{t}\le\frac{\lambda\log(2d)}{T}+32\sqrt{d\tau^{2}/T}+128d\tau^{2}/\lambda
\]
The initialization event has probability at least $1/12$, and the
noise event has conditional probability at least $3/4$ for every
such initialization. Their intersection therefore has probability
at least $1/16$. For any execution in this intersection, nonnegativity
of the $\Delta_{t}$'s implies that a uniformly chosen iterate has
error at most twice their average with probability at least $1/2$.
This gives the claimed $1/32$ success probability.
\end{proof}

\begin{cor}
\label{cor:noisy} There is a universal sufficiently large constant
$C$ such that the following holds. For $0<\alpha<1$, set $T=\left\lceil C\frac{\log(2d)}{\alpha}\right\rceil $.
If the noise satisfies $\frac{d\tau^{2}T}{\lambda^{2}}\le c_{0}=2^{-14}$
then a uniform output iterate satisfies $x^{\top}_{J}Wx_{J}\ge(1-\alpha)\lambda$
with probability at least $1/32$. 
\end{cor}

\begin{proof}
The noise condition makes $\sqrt{d\tau^{2}/T}\le\sqrt{c_{0}}\lambda/T$
and $d\tau^{2}/\lambda\le c_{0}\lambda/T$. By applying Theorem~\ref{thm:noisy},
we obtain 
\[
\frac{1}{T}\sum_{t<T}\Delta_{t}\le\lambda\cdot\left(\frac{\alpha}{C}+\frac{1}{4}\frac{\alpha}{C\log(2d)}+\frac{\alpha}{C\log(2d)}\right).
\]
The corollary follows from setting a sufficiently large $C$.
\end{proof}

\section{Proof of the main theorem}\label{sec:Proof-main-theorem}

We are now ready to prove Theorem \ref{thm:main}.
\begin{proof}[Proof of Theorem \ref{thm:main}]
 Privacy follows from Proposition~\ref{prop:privacy} on arbitrary
datasets. Next, we analyze the utility. Define 
\[
\whSigma=\frac{1}{n}\sum^{n}_{i=1}y_{i}y^{\top}_{i},\qquad\widehat{\lambda}=\|\whSigma\|_{2}.
\]
Take $T=\lceil C\log(2d)/\alpha\rceil$, and choose $R$ and $\tau$
as in Propositions~\ref{prop:no-clipping} and~\ref{prop:privacy}.
Set $\gamma=1/256$. Thus 
\begin{equation}
R=\Theta(\Lambda\left(\sqrt{du}+u\right)),\qquad\tau=\Theta\left(\frac{R\sqrt{T\log(1/\bar{\delta})}}{n\eps}\right),\qquad d\tau^{2}=\Theta\left(\frac{\Lambda^{2}(d+u)udT\log(1/\bar{\delta})}{n^{2}\eps^{2}}\right).\label{eq:V}
\end{equation}
Lemma \ref{lem:cov} gives 
\begin{equation}
\|\whSigma-\Sigma\|_{2}\le c\sqrt{d/n}\lambda_{1}\le c\alpha\lambda_{1},\qquad\tfrac{1}{2}\lambda_{1}\le\widehat{\lambda}\le2\lambda_{1}\label{eq:cov-good}
\end{equation}
with probability at least $1-\gamma$ if $n\ge C(d+\log(1/\gamma))/\alpha^{2}$.
On this event, $\Lambda\le4\widehat{\lambda}$, so 
\[
\frac{d\tau^{2}T}{\widehat{\lambda}^{2}}=\Theta\left(\frac{(d+u)udT^{2}\log(1/\bar{\delta})}{n^{2}\eps^{2}}\right).
\]
Consequently a sample bound of $n=\Omega(T\sqrt{du(d+u)\log(1/\bar{\delta})}/\eps)=\Omega(\log(2d)\sqrt{du(d+u)\log(1/\bar{\delta})}/(\alpha\eps))$,
with sufficiently large constants, ensures the noise condition $\frac{d\tau^{2}T}{\widehat{\lambda}^{2}}\le c_{0}=2^{-14}$
in Theorem~\ref{thm:noisy}. Its three error terms are 
\[
\frac{\widehat{\lambda}\log(2d)}{T}=O(\alpha\lambda_{1}),\quad\sqrt{\frac{d\tau^{2}}{T}}=O\left(\frac{\lambda_{1}\sqrt{du(d+u)\log(1/\bar{\delta})}}{n\eps}\right),\quad\frac{d\tau^{2}}{\widehat{\lambda}}=O\left(\frac{\lambda_{1}du(d+u)\log(2d)\log(1/\bar{\delta})}{\alpha n^{2}\eps^{2}}\right).
\]
The same sample bound makes the latter two terms $O(\alpha\lambda_{1})$.

The total error is 
\[
\lambda_{1}-x^{\top}\Sigma x\le(\widehat{\lambda}-x^{\top}\widehat{\Sigma}x)+2\left\Vert \whSigma-\Sigma\right\Vert _{2}
\]
where the first term is bounded by Theorem \ref{thm:noisy} and the
second term is bounded by Lemma \ref{lem:cov}.

Thus, for $n=\Theta\left(\frac{\log(2d)\sqrt{du(d+u)\log(1/\bar{\delta})}}{\alpha\eps}+\frac{d+1}{\alpha^{2}}\right)$
samples, Theorem~\ref{thm:noisy} bounds the empirical error of the
ideal non-clipping algorithm's output by $O(\alpha\lambda_{1})$ with
probability at least $1/32$. Clipping happens with probability at
most $\gamma$. When clipping does not happen, the actual algorithm
and the ideal coincide.

Lemma \ref{lem:cov} shows that the empirical covariance is close
to the distributional one with probability $1-\gamma$. Thus, the
probability that the empirical covariance is accurate, non-clipping
and the algorithm succeeds is at least $\frac{1}{32}-\gamma-\gamma\ge1/64$.
\end{proof}

\bibliographystyle{plain}
\bibliography{gaussian-pca}

\begin{thebibliography}{10}

\bibitem{AllenZhuLi2017}
Zeyuan Allen-Zhu and Yuanzhi Li.
\newblock First efficient convergence for streaming {$k$}-{PCA}: A global,
  gap-free, and near-optimal rate.
\newblock In {\em 2017 IEEE 58th Annual Symposium on Foundations of Computer
  Science (FOCS)}, pages 487--492. IEEE, 2017.

\bibitem{DBLP:conf/colt/BalcanDWY16}
Maria{-}Florina Balcan, Simon~Shaolei Du, Yining Wang, and Adams~Wei Yu.
\newblock An improved gap-dependency analysis of the noisy power method.
\newblock In {\em {COLT}}, volume~49 of {\em {JMLR} Workshop and Conference
  Proceedings}, pages 284--309. JMLR.org, 2016.

\bibitem{DavidsonSzarek01}
Kenneth~R. Davidson and Stanis{\l}aw~J. Szarek.
\newblock Local operator theory, random matrices and banach spaces.
\newblock In {\em Handbook of the Geometry of Banach Spaces}, volume~1, pages
  317--366. Elsevier, 2001.

\bibitem{dOrsiN26}
Tommaso d'Orsi and Gleb Novikov.
\newblock Tight differentially private {PCA} via matrix coherence.
\newblock In Kasper~Green Larsen and Barna Saha, editors, {\em Proceedings of
  the 2026 Annual {ACM-SIAM} Symposium on Discrete Algorithms, {SODA} 2026,
  Vancouver, BC, Canada, January 11-14, 2026}, pages 10--51. {SIAM}, 2026.

\bibitem{DBLP:conf/stoc/DworkTT014}
Cynthia Dwork, Kunal Talwar, Abhradeep Thakurta, and Li~Zhang.
\newblock Analyze {Gauss}: optimal bounds for privacy-preserving principal
  component analysis.
\newblock In {\em {STOC}}, pages 11--20. {ACM}, 2014.

\bibitem{GuKumarTianYang2026}
Anming Gu, Syamantak Kumar, Kevin Tian, and Chutong Yang.
\newblock Gap-free streaming {PCA} beyond rank-one updates: Near-optimal rates
  and applications to differential privacy, 2026.
\newblock arXiv:2609.26508.

\bibitem{DBLP:conf/nips/HardtP14}
Moritz Hardt and Eric Price.
\newblock The noisy power method: {A} meta algorithm with applications.
\newblock In {\em {NIPS}}, pages 2861--2869, 2014.

\bibitem{DBLP:conf/stoc/HardtR12}
Moritz Hardt and Aaron Roth.
\newblock Beating randomized response on incoherent matrices.
\newblock In {\em {STOC}}, pages 1255--1268. {ACM}, 2012.

\bibitem{DBLP:conf/stoc/HardtR13}
Moritz Hardt and Aaron Roth.
\newblock Beyond worst-case analysis in private singular vector computation.
\newblock In {\em {STOC}}, pages 331--340. {ACM}, 2013.

\bibitem{KarwaVadhan2018}
Vishesh Karwa and Salil Vadhan.
\newblock Finite sample differentially private confidence intervals.
\newblock In {\em 9th Innovations in Theoretical Computer Science Conference
  (ITCS 2018)}, volume~94 of {\em Leibniz International Proceedings in
  Informatics (LIPIcs)}, pages 44:1--44:9. Schloss Dagstuhl -- Leibniz-Zentrum
  f{\"u}r Informatik, 2018.

\bibitem{LaurentMassart2000}
B{\'e}atrice Laurent and Pascal Massart.
\newblock Adaptive estimation of a quadratic functional by model selection.
\newblock {\em The Annals of Statistics}, 28(5):1302--1338, 2000.

\bibitem{DBLP:conf/nips/LiuK0O22}
Xiyang Liu, Weihao Kong, Prateek Jain, and Sewoong Oh.
\newblock {DP-PCA:} statistically optimal and differentially private {PCA}.
\newblock In {\em NeurIPS}, 2022.

\bibitem{liu2022differential}
Xiyang Liu, Weihao Kong, and Sewoong Oh.
\newblock Differential privacy and robust statistics in high dimensions.
\newblock In {\em {COLT}}, pages 1167--1246. PMLR, 2022.

\bibitem{LiuNNEN23}
Zijian Liu, Ta~Duy Nguyen, Thien~Hang Nguyen, Alina Ene, and Huy Nguyen.
\newblock High probability convergence of stochastic gradient methods.
\newblock In {\em {ICML}}, volume 202 of {\em Proceedings of Machine Learning
  Research}, pages 21884--21914. PMLR, 23--29 Jul 2023.

\bibitem{NguyenEN26}
Ta~Duy Nguyen, Alina Ene, and Huy~Le Nguyen.
\newblock Adaptive power iteration method for differentially private {PCA}.
\newblock {\em CoRR}, abs/2602.11454, 2026.

\bibitem{NicolasSMMC25}
Julien Nicolas, César Sabater, Mohamed Maouche, Sonia~Ben Mokhtar, and Mark
  Coates.
\newblock Differentially private and decentralized randomized power method,
  2025.

\bibitem{Shamir2016PCA}
Ohad Shamir.
\newblock Convergence of stochastic gradient descent for {PCA}.
\newblock In {\em Proceedings of the 33rd International Conference on Machine
  Learning}, volume~48 of {\em Proceedings of Machine Learning Research}, pages
  257--265. PMLR, 2016.

\end{thebibliography}

\end{document}